\documentclass[a4paper,UKenglish,cleveref, autoref, thm-restate]{lipics-v2021}

\newcommand{\Vect}{{\rm Vec}}
\newcommand{\T}{{\cal T}}
\newcommand{\true}{\mbox{\sc True}}
\newcommand{\false}{\mbox{\sc False}}
\newcommand{\kibitz}[2]{\ifnum\Comments=1{\color{#1}{#2}}\fi}

\renewcommand{\paragraph}[1]{\medskip\noindent{\bf #1}}

\title{Collision Detection for Modular Robots -- it is easy to cause collisions and hard to avoid them} 

\titlerunning{Collision Detection for Modular Robots} 

\author{Siddharth Gupta}{University of Warwick, United Kingdom}{siddharth.gupta.1@warwick.ac.uk}{}{}

\author{Marc van Kreveld}{Utrecht University, The Netherlands}{m.j.vankreveld@uu.nl}{}{}

\author{Othon Michail}{University of Liverpool, United Kingdom}{othon.michail@liverpool.ac.uk}{}{}

\author{Andreas Padalkin}{Paderborn University, Germany}{andreas.padalkin@upb.de}{}{}

\authorrunning{Gupta, van Kreveld, Michail, Padalkin} 

\Copyright{Gupta, van Kreveld, Michail, Padalkin} 



\acknowledgements{The authors thank all participants of the Bertinoro Workshop on Distributed Geometric Algorithms, in particular Peyman Afshani for suggesting the $O(n\log^2 n)$ time solution for detecting collisions when there are no couplings. We thank Irina Kostitsyna and Christian Scheideler for the organization, and the latter also for proposing the collision detection problem. Finally, we thank Jesper Nederlof for some useful observations.}

\nolinenumbers 

\begin{document}

\maketitle

\begin{abstract}
We consider geometric collision-detection problems for modular reconfigurable robots. Assuming the nodes (modules) are connected squares on a grid, we investigate the complexity of deciding whether collisions may occur, or can be avoided, if a set of expansion and contraction operations is executed.
We study both discrete- and continuous-time models, and allow operations to be coupled into a single parallel group. Our algorithms to decide if a collision may occur run in $O(n^2\log^2 n)$ time, $O(n^2)$ time, or $O(n\log^2 n)$ time, depending on the presence and type of coupled operations, in a continuous-time model for a modular robot with $n$ nodes. To decide if collisions can be avoided, we show that a very restricted version is already NP-complete in the discrete-time model, while the same problem is polynomial in the continuous-time model. A less restricted version is NP-hard in the continuous-time model.
\end{abstract}

\section{Introduction}

Modular reconfigurable robotics and the related concept of programmable matter concern systems composed of interconnected elementary entities, called modules. 
The collection of modules can coordinate its limited communication, computation, sensing, and local actuation
to accomplish nontrivial global tasks. Local actuation of modules is enabled through a set of one or more mechanical operations that they can perform. An operation typically involves the module that applies it as well as modules in its local neighborhood. Examples of such operations are pushing, pulling, expanding, contracting, doubling, and rotating. There are systems and models using operations that can have a direct effect on the global structure. This is particularly true under the combined effect of several such operations overlapping in time.

The ability of local operations to globally affect the robotic structure is a double-edged sword. On one hand, it is a convenient form of parallelism, where global structural changes can happen faster. On the other hand, if not properly orchestrated, it could cause small violations of the structure or even complete structural failure. We, hereafter, shall call all structural violations and failures \emph{collisions}. Operations that---when applied on individual modules---can globally affect the structure, are sometimes called \emph{linear-strength operations}.

The positive effect of such operations has been studied from a theoretical point of view in a number of papers, for different underlying models and types of operations. In a series of papers, Aloupis \emph{et al.} \cite{aloupis2008reconfiguration,aloupis2009linear,aloupis2009realistic} studied a model of a robotic system known as crystalline robots \cite{rus2001crystalline}.
The 2D version of the crystalline model represents modules as squares on a 2D grid, forming a connected shape of modules attached to adjacent modules. Each individual module can expand and contract, by extending one of its faces one unit out and retracting it back at some later point. Due to modules being attached to each other, up to linear-size components can move due to a module's expansion or contraction.
In \cite{aloupis2008reconfiguration}, Aloupis \emph{et al.} gave a universal centralized reconfiguration algorithm for the crystalline model that, for any pair of connected shapes $S_I, S_F$ of the same number of modules $n$, can transform $S_I$ into $S_F$ within $O(\log n)$ parallel time steps by performing $\Theta(n\log n)$ individual operations.

In \cite{woods2013active}, Woods \emph{et al.} proposed the nubot model, motivated by the programmable self-assembly of molecules, such as DNA strands. In this model, modules represent monomers on a 2D triangular grid. The model incorporates a number of different types of operations, such as insertion, deletion, and rotation of modules.
The motion caused by operations is propagated to larger parts of the shape through its connections. Operations whose global effect would violate the rigidity of a connection are assumed to be canceled. Their main result is a distributed, asynchronous algorithm
which, starting from a singleton, can form any connected 2D shape and pattern of size $n$, within a polylogarithmic (in $n$) number of parallel time steps in expectation.

Almalki and Michail \cite{almalki2022geometric}, building on the insertion operations of \cite{woods2013active} and the growth processes on graphs by Mertzios \emph{et al.} \cite{mertzios2022complexity}, investigated what families of shapes can be grown in time polylogarithmic in their size by using only growth operations.
They did this in a 2D square grid model, under different requirements on operation couplings that must be satisfied in each time step.
Their operations were defined so that collisions can never occur. They gave centralized algorithms for growing a shape $S_F$ from a shape $S_I$ (possibly a singleton), which yield polylogarithmic parallel time-step schedules for large classes of shapes.

The amoebot model of Derakhshandeh \emph{et al.} \cite{derakhshandeh2014amoebot} --and its recent canonical extension \cite{daymude2021canonical}-- is another model in which the main operations considered are expansions and contractions of modules. The modules operate on a 2D triangular grid and reconfiguration happens through expansions of the head of modules toward empty space, followed later by contraction of their tail toward their head.
Shape formation algorithms in this model are usually designed in a way that operations are parallel but each is affecting only a local region around it and not larger parts of the shape. See \cite{daymude2019computing}, for a recent chapter covering the main algorithmic developments in this model. Recently, Feldmann \emph{et al.} \cite{DBLP:journals/jcb/FeldmannPSD22} have proposed to add linear-strength operations to the model,
but they have left the details of such an extension for future work.

The necessity to avoid collisions between large moving parts of a shape is also present in programmable matter by folding \cite{hawkes2010programmable} and the algorithmic questions related to it \cite{demaine2017origamizer}. In those models, the goal is to fold a polyhedral complex starting from a sufficiently large surface, such as a piece of paper,
without self-intersecting in the process. Coordinated motion planning is another line of research having some similarities to our work. There the goal is to reconfigure a swarm of robots to a target configuration as fast as possible while avoiding collisions of robots
\cite{schwartz1983piano,demaine2019coordinated}. Another model in which individual operations can have up to a linear effect on the shape, is the line-pushing model of Almethen \emph{et al.} \cite{almethen2020pushing}.

It is evident that most studies have restricted attention to those operations that are safe to perform in parallel. These are either linear-strength operations that cannot collide or operations that affect only the local region around them. In this paper, we explicitly pose the algorithmic question of determining when a set of operations may cause a collision and when a collision can be avoided. In particular, given a shape and a set of linear-strength operations on that shape 
we aim to give centralized algorithms that can compute a schedule of these (sets of) operations that would (i) cause a collision or (ii) avoid collisions. The former subquestion is motivated by asynchronous distributed algorithms, in which any of the possible interleavings of operations might be the one that the modules will actually realize; the latter by the need to design efficient reconfiguration algorithms that avoid collisions, instead of having collision-avoidance built in the model. To the best of our knowledge, the present is the first study to be explicitly considering this type of questions.

\begin{figure}[!hbtp]
\centering
\begin{subfigure}[htb]{0.44\textwidth}
	\centering
	\includegraphics[width=\textwidth]{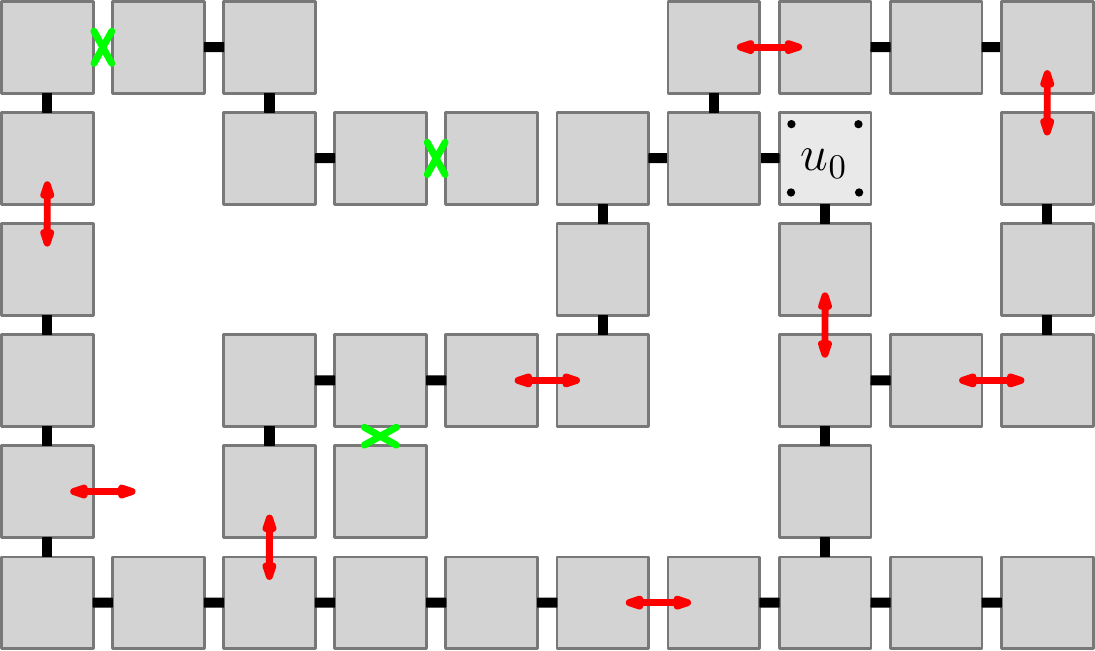}
	\label{fig:example-model-operations}
\end{subfigure}
\hfill
\begin{subfigure}[htb]{0.508\textwidth}
	\centering\vspace{0.52cm}
	\includegraphics[width=\textwidth]{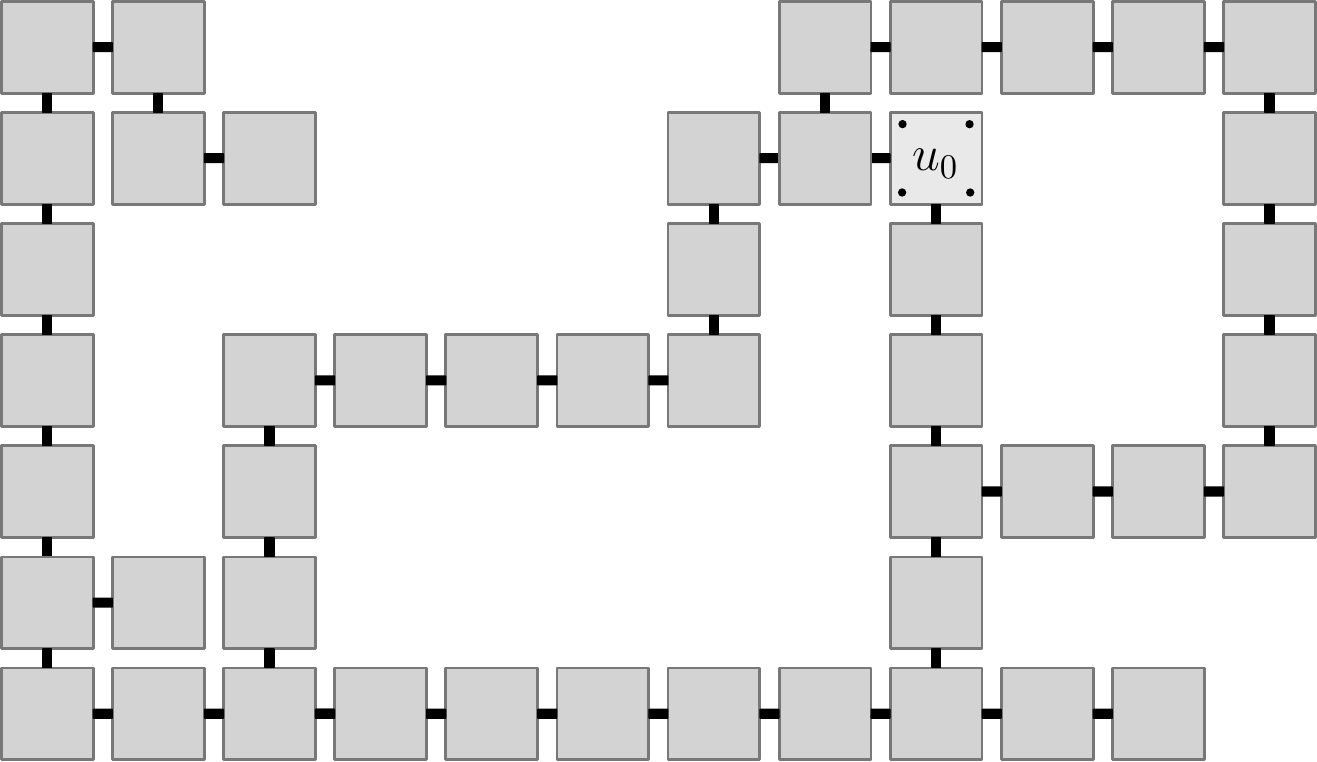}
	\label{fig:example-model-outcome}
\end{subfigure}
\caption{Applying a collision-free set of parallel operations to a shape.
	In general, a shape's edge set is shown by the black edges between the nodes, and $u_0$ denotes the anchor. Red double arrows indicate an expansion and green X's indicate a contraction. Any operation applied between two nodes should be regarded as taking place over an edge, which is not drawn to keep the illustration clean. For the same reason, nodes will be drawn slightly smaller than their actual size.}
\label{fig:example-model}
\end{figure}

\paragraph{Approach.}
We choose to study these questions in a 2D square grid model, where modules, called nodes hereafter, are unit squares occupying distinct cells of the grid. This choice makes the 2D crystalline model of Aloupis \emph{et al.} \cite{aloupis2008reconfiguration} and the growth model of Almalki and Michail \cite{almalki2022geometric} to be the closest to our model. Nodes can be connected to some of their adjacent nodes, in a way that always forms an initial connected shape $S$. We do not allow new connections to be created between the nodes. 
The operations considered are expansion and contraction.
These operations can also be viewed as a linear-strength extension of the expansions and contractions of the amoebot model \cite{derakhshandeh2014amoebot}. The input to our problems is a connected shape $S$, an assignment of operations on $S$ without multiplicities,
and a coupling partition of the operations. The coupling partition specifies which operations will be performed in parallel. See Figure~\ref{fig:example-model} for an example.

\paragraph{Results and overview.}
Section~\ref{sec:model} describes the model in which we study the collision problems. In particular, we specify the assumed input, the operations and their coupling, collisions, and time models. Then we give our problem definitions.
Section~\ref{section:collisions} first presents algorithms for the problem of detecting whether collisions may occur in the continuous-time model, where operations may start any time. When there is no specified parallelity of operations, we give an $O(n\log^2 n)$ time algorithm. When certain operations are specified to be in parallel, we present algorithms that run in $O(n^2)$ time or $O(n^2\log^2 n)$ time, depending on the degree of parallelity. We then give polynomial-time algorithms in the discrete-time model, where any operation or coupled group must complete its action before the next one begins.
Section~\ref{sec:avoiding} addresses the problem of determining whether all operations can be performed while avoiding collisions. We show that a restricted version is already NP-complete in the discrete-time model, while it is polynomial-time solvable in the continuous-time model. A somewhat less restricted version is again NP-hard in the continuous-time model. The NP-hardness results require that certain operations be specified as parallel.

\section{Model}
\label{sec:model}

We assume a 2-dimensional square grid where each cell has integer coordinates $(x,y)$. Nodes (modules) occupy cells, defining a set of occupied integer points 
such that no two nodes occupy the same cell. We represent every node $u=(u_x,u_y)$ as a square of size equal to and perfectly aligned with cell $(u_x,u_y)$ of the grid.
A \emph{shape} $S=(V,E)$ is a configuration of nodes $V$ together with their connectivity, represented by $E$. Only orthogonally adjacent nodes can be connected, but adjacent nodes are not necessarily connected.
We use $n$ to denote $|V|$ and restrict our attention to connected shapes, throughout.

\paragraph{Operations and collisions.}
In general, applying one or more \emph{operations} to a shape $S$ either causes a \emph{collision} or yields a new shape $S'$. Collisions come in two types: \emph{node collisions} and \emph{cycle collisions}. 
Given that all collisions here will be ``self-collisions'' of a connected shape, we can assume without loss of generality (abbreviated ``w.l.o.g.'' throughout) that there is an \emph{anchor} node $u_0\in V$ that is stationary and other nodes move relative to it.
We begin with the simpler case where the shape is a tree $T=(V,E)$, where cycle collisions do not exist, and then generalize to any connected shape $S$.

We start by defining single \emph{expansion}
and \emph{contraction}
operations;\footnote[1]{We believe that our definitions and techniques can be extended to alternative versions of expansion and contraction---including the case where the operations can be reversed---and to different geometries such as a triangular grid.}
see Figure~\ref{fig:example-model}. An \emph{expansion} operation is applied to a pair of adjacent integer points $uv$, where either (i) $u\in V$ and $v\notin V$, or (ii) $u,v\in V$ and $uv\in E$ holds.
The remaining case where $u,v\in V$ but $uv\notin E$ immediately gives a collision.
In case~(i), the expansion generates a node at the empty cell $v$ connected to $u$. In case~(ii), assume w.l.o.g. that $u$ is closer to $u_0$ in $T$ than $v$.
Let $T(v)$  denote the subtree of $T$ rooted at $v$. Then, the expansion generates a node between $u$ and $v$, connected to both, which translates $T(v)$ by one unit away from $u$ along the axis parallel to~$uv$. In both cases, the new node starts as a unit-length segment that widens into a unit square. 
A \emph{contraction} operation is applied to a pair of nodes $uv\in E$, $v$ being the furthest from the anchor. 
It merges $v$ with $u$ by translating $T(v)$ by one unit toward $u$ while $v$ narrows to a unit-length segment.
In both types of operations, if after $T(v)$'s translation two nodes occupy the same cell then a collision has occurred. We call this type of collision a \emph{node collision} and more generally define it as the non-empty intersection of the areas of any two nodes at any point in time. Otherwise, a new tree $T'$ 
has been obtained.

We assume that no node is involved in more than one operation.

\paragraph{Coupling.}
Let $Q$ be a set of operations to be applied \emph{in parallel} to a connected shape $S$, each operation on a distinct pair of nodes or a node and an unoccupied cell. 
We call such a set a \emph{coupling}, and the operations it contains are \emph{coupled} or \emph{parallel}. 
We assume that all operations in $Q$ are applied \emph{concurrently}, have the same \emph{constant execution speed}, and their \emph{duration} is equal to one unit of time.

Let $T=(V,E)$ be a tree and $u_0\in V$ its anchor. We set $u_0$ to be the root of $T$. We want to determine the displacement of every $v\in V\setminus\{u_0\}$ due to the parallel application of the operations in $Q$. As $u_0$ is stationary and each operation translates a subtree, only the operations on the unique $u_0v$ path contribute to $v$'s displacement.
In particular, any such operation contributes one of the unit vectors $\langle -1,0\rangle, \langle 0,-1\rangle, \langle +1,0\rangle, \langle 0,+1\rangle$ to the motion vector $\vec{v}$ of~$v$.
Moreover, for any node $u\in V$ that expands toward an empty cell, we add a new node $v$ with a corresponding unit motion vector $\vec{v}$.
We can use the set of motion vectors to determine whether the trajectories of any two nodes will collide at any point. 
Let now $S$ be any connected shape with at least one cycle and any node $u_0$ be its anchor. Then, 
a set of operations $Q$ on $S$ either causes a \emph{cycle collision} or its effect is essentially equivalent to the application of $Q$ on any spanning tree of $S$ rooted at $u_0$. 
Let $u$, $v$ be any two nodes on a cycle. If $p_1$ and $p_2$ are the two $uv$ paths of the cycle, then $\vec{v}_{p_1}=\vec{v}_{p_2}$ must hold:
the displacement vectors along the paths $p_1$ and $p_2$ are equal. 
Otherwise, we cannot maintain all nodes or edges of the cycle. Such a violation is called a \emph{cycle collision}. We call a set of operations that does not cause any node or cycle collisions  \emph{collision free}.

\paragraph{Discrete and continuous time.}
We consider two different models for the scheduling of the operations.
In the \emph{discrete-time model}, each operation or coupling starts at a different integer time and takes one unique unit of time. In other words, no two operations are active at the same time unless they are coupled.
In the \emph{continuous-time model}, we do not make the integer starting-time assumption. Operations can start at any time and their active times can overlap. Coupled operations start and finish at the same time.
Our assumption that each operation takes one unit of time to complete and has constant execution speed holds for both timing models. In the discrete-time model, only the order of the operations (individual or coupled) matters for having collisions or not. In the continuous-time model, the precise starting times of the operations matter.

\paragraph{Problem definitions.}
We now define the problems considered. Given a shape $S$ and an assignment of operations on $S$ that involve any node at most once, a \emph{coupling partition of operations on $S$} is a collection of sets $\{Q_1,Q_2,\ldots,Q_k\}$, where each $Q_i$ (possibly a singleton) denotes a subset of the operations that should be performed in parallel.
\smallskip

\noindent\textbf{{\sc Colliding Schedule}.} Given a shape $S=(V,E)$ from a given family of shapes and a coupling partition of operations $\{Q_1,Q_2,\ldots,Q_k\}$ on $S$, decide if a starting time $t_0(Q_i)\in \mathbb{R}$ for each coupled set $Q_i$ exists such that the application of the operations according to these starting times causes a collision.
\smallskip

\noindent\textbf{{\sc Collision-free Schedule}.} Given a shape $S=(V,E)$ from a given family of shapes and a coupling partition of operations $\{Q_1,Q_2,\ldots,Q_k\}$ on $S$, decide if a starting time $t_0(Q_i)\in \mathbb{R}$ for each coupled set $Q_i$ exists such that the application of the operations according to these starting times is collision free.
\smallskip

The discrete special cases of these problems, {\sc Discrete Colliding Schedule} and {\sc Discrete Collision-free Schedule}, respectively, are obtained by requiring all $t_0(Q_i)$'s to be unique integers.

\section{Algorithms for Colliding Schedule}
\label{section:collisions}

In this section, we present algorithms to decide whether a connected shape can have collisions for some schedule of the operations. We distinguish in the case without couplings and cases with various forms of coupling. We first discuss the case where $S$ is a tree; then we extend to solve {\sc Colliding Schedule} for general connected graphs. 

\subsection{Continuous and Discrete Colliding Schedule for Trees}
\label{subsection:collisions-trees}

We assume that the topology of $S$ is that of a tree.
We first present a general method to decide if collisions may occur which works when couplings may exist. The method is slightly more efficient when all couplings have constant size or each coupling is horizontal-only or vertical-only.  Then we present an improved algorithm for the case when no couplings exist.

Let $v$ be any node in $S$. We will give an algorithm to decide if $v$ can collide with any other node $w$.  
Recall that there is exactly one path between $v$ and $w$, and only the operations and couplings on this path determine whether $v$ and $w$ can collide.
Assuming that $v$ is stationary, coupled operations on this path allow $w$ to translate over a vector $\langle i,j\rangle$ with integer coordinates. By letting all nodes take the role of $v$, we get a complete algorithm.

\begin{figure}[tb]
\centering\includegraphics[width=0.9\textwidth]{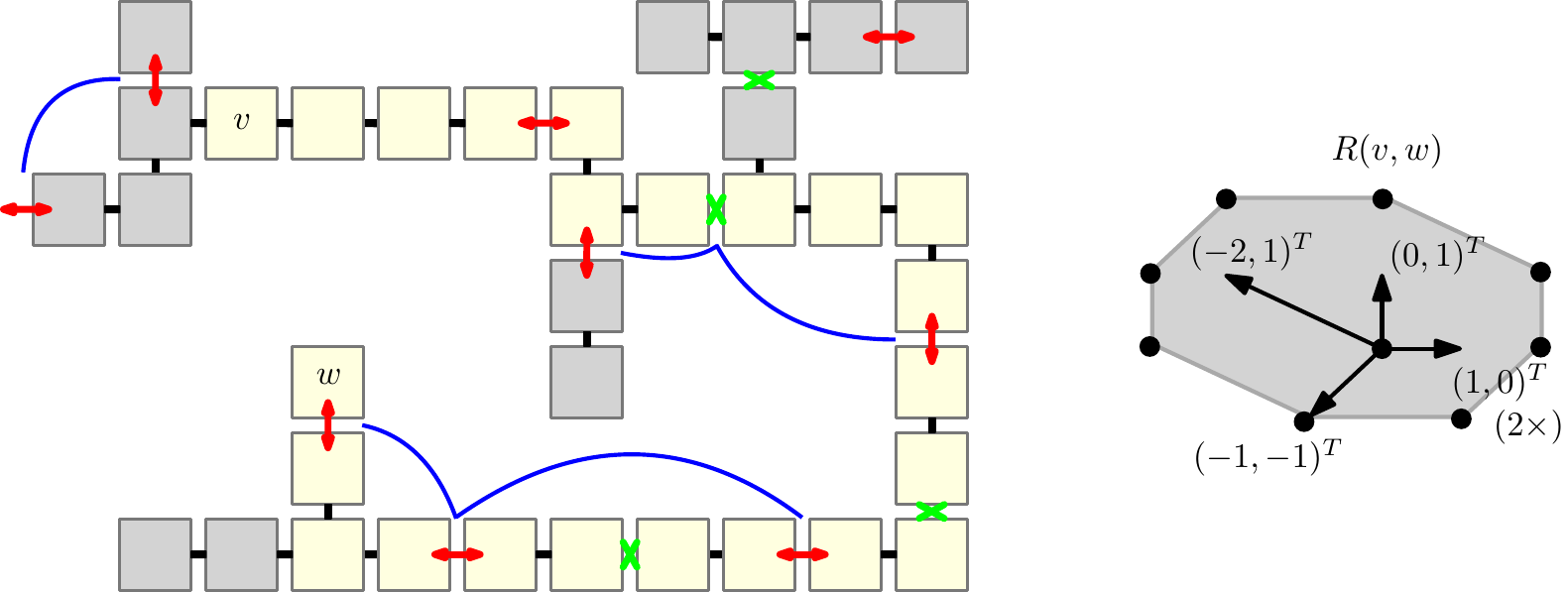}
\caption{A shape $S$ with operations and $v$ and $w$ indicated. Only operations with both ends on the $vw$ path influence how $w$ can move with respect to $v$. The couplings are shown with blue arcs. There are five vectors in total implied by the path, of which two are the same. The corresponding region $R(v,w)$ for $w$ with respect to $v$ is the zonotope shown on the right. The $(2\times)$ label in the zonotope represents the fact that the vector $(1,0)^T$ is implied twice by the path.}
\label{fig:path}
\end{figure}

Starting at $v$, we perform a traversal of $S$ and visit all other nodes. Whenever we visit and treat a node $w'$, we have already visited the nodes on the path from $v$ to $w'$, in particular the neighbor $w$ of $w'$ on this path. The idea of the algorithm is to maintain the vectors made by the operations and couplings, with multiplicity, when traversing $S$ from node to node.

Let $\Vect(v,w)$ be the multi-set of vectors describing the independent operations and couplings between $v$ and $w$ (see Figure~\ref{fig:path}), and let $\T(v,w)$ be the tree storing $\Vect(v,w)$ sorted by angle in the leaves. We augment $\T(v,w)$ by storing at every internal and leaf node the sum vector of all vectors (with multiplicities) in the leaves below it. 

One extra traversal step in $S$, going from $w$ to $w'$, means we must update $\T(v,w)$ to obtain $\T(v,w')$:
$(i)$ If the edge between $w$ and $w'$ is not an operation, then $\T(v,w')=\T(v,w)$. $(ii)$ If the edge is an operation without coupling on the path so far, we get one new vector from $\langle -1,0\rangle, \langle 0,-1\rangle, \langle +1,0\rangle, \langle 0,+1\rangle$, which we insert into $\T(v,w)$ (possibly by increasing a multiplicity count instead of a leaf insertion) to create $\T(v,w')$.
$(ii)$ If the edge is an operation coupled with one or more operations between $v$ and $w$, then a vector is modified by adding an extra operation to it. The corresponding vector changes by $\pm 1$ in the $x$- or $y$-direction. The vector can become the zero vector, in which case we remove it. In the other case, we perform a deletion and an insertion to change the vector.

The tree $\T(v,w)$ is an implicit representation of the reachable region $R(v,w)$ of a node $w$ with respect to $v$. In the continuous-time model,
the reachable region is a convex region that contains the node $w$ itself. It is the \emph{zonotope} of the vectors $\Vect(v,w)$, and has twice as many edges as there are unique directions of vectors (we note that $R(v,w)$ is not convex if there are couplings and we use the discrete-time model). 
With a continuous-time model, $w$ can always intersect any cell inside by starting the operations at suitable times.

The next lemma shows that we can query the search tree $\T(v,w)$ to find an extreme point of $R(v,w)$ in a given query direction efficiently, and test for intersection of nodes $v$ and $w$.

\begin{lemma}
\label{lem:tree}
We can find an extreme point of $R(v,w)$ in a given query direction $\vec{\rho}$ in $O(\log n)$ time, and test whether $v$ and $w$ can intersect in $O(\log^2 n)$ time.
\end{lemma}

\begin{proof}
Let $\vec{\rho}$ be any vector; we will use its direction only. We explain how to use $\T(v,w)$ to find a vertex of $R(v,w)$
that is extreme in the direction of $\vec{\rho}$. For such a vertex $u$, a line through $u$ normal to $\vec{\rho}$ will have the interior of $R(v,w)$ strictly to one side. To find this vertex, we add up all vectors that have a positive dot product with $\vec{\rho}$. The augmentation of $\T(v,w)$ allows us to find this vertex as the addition of $O(\log n)$ vectors stored at internal or leaf nodes. We follow the search paths down $\T(v,w)$ using the two directions normal to $\vec{\rho}$, being $\vec{\rho}-90^\circ$ and $\vec{\rho}+90^\circ$. Suppose we end in leaves $\nu$ and $\nu'$, respectively. For all highest nodes between the paths to $\nu$ and $\nu'$, we add up the summed vectors stored in those nodes.
(Note that we may need the nodes left of the left path and right of the right path, since the directions are a cyclic order. We ignore this fully analogous case.)
If appropriate, we also add the vectors in one or both of $\nu$ and $\nu'$ (with their multiplicities). The sum gives the extreme vertex $u$ of $R(v,w)$ in direction $\vec{\rho}$ in $O(\log n)$ time.
When we turn $\vec{\rho}$ a little, $u$ will remain the extreme vertex, until the direction changes enough that the set of vectors with a
positive dot product with $\vec{\rho}$ changes. This happens when there is a vector in direction $\vec{\rho}-90^\circ$ or $\vec{\rho}+90^\circ$ in $\Vect(v,w)$. Such a vector is stored in $\nu$ or $\nu'$ or an adjacent leaf. By examining them, we can determine an adjacent extreme vertex $u'$ of $R(v,w)$, and hence an edge $uu'$ that bounds $R(v,w)$ in $O(\log n)$ time as well.

\begin{figure}[htb]
\centering\includegraphics{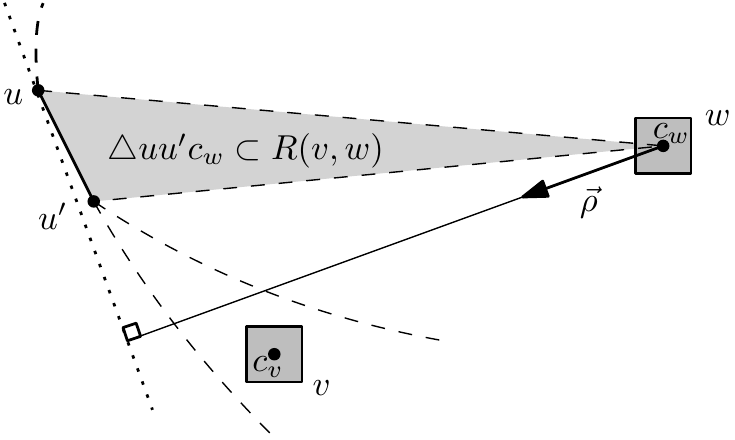}
\caption{Illustration of case $(iii)$, where $u$ is an extreme vertex of $R(v,w)$ in direction $\vec{\rho}$, and the locations of $u$ and $u'$ do not tell yet whether $c\in R(v,w)$ (the dashed curves show possible continuations of $R(v,w)$ counterclockwise). Here we must search more counterclockwise by choosing another $\vec{\rho}$ positively rotated with respect to the current one.}
\label{fig:reachable-region}
\end{figure}

We will use the query for an extreme vertex of $R(v,w)$ in a direction $\vec{\rho}$ to 
determine whether $v$ intersects the region that $w$ can reach. 
We first show how to decide whether the center point $c_v$ of node $v$ lies inside $R(v,w)$, the region that the center $c_w$ of node $w$ can reach; then we adapt this to the cells themselves. We decide if $c_v\in R(v,w)$ by finding the edge $uu'$ of $R(v,w)$ for which $\overrightarrow{c_wc_v}$ is wedged between $\overrightarrow{c_wu}$ and $\overrightarrow{c_wu'}$, see Figure~\ref{fig:reachable-region}. Once we have this edge, we can test whether $c_v \in \triangle uu'c_w$, in which case we found a collision, and otherwise, we did not.

To find the edge $uu'$ of $R(v,w)$, we perform a binary search, starting with the angular interval for $\vec{\rho}$ with counterclockwise boundary $\overrightarrow{c_wc_v} - 90^\circ$ and clockwise boundary $\overrightarrow{c_wc_v} + 90^\circ$. 
The two boundaries lead to leaves in $\T(v,w)$. We consider the leaf $\mu$ halfway in between to perform one binary search step. This leaf $\mu$ corresponds to a direction $\vec{\rho}$ from $w$, and this direction gives a pair of extreme vertices $u$ and $u'$ of $R(v,w)$.
There are three cases, see Figure~\ref{fig:reachable-region}. 
$(i)$ If $c_v$ is inside the triangle $\triangle uu'c_w$ then we know that $c_v$ is inside $R(v,w)$ and we found a possible collision.
Otherwise, $(ii)$ if $c_v$ is on the other side of the line through $u$ and $u'$ than $c_w$, then $c_v$ cannot be inside $R(v,w)$ by convexity, and we can stop.
Otherwise, $(iii)$ we must inspect directions more clockwise or more counterclockwise from $\vec{\rho}$.

With this binary search, we need at most $O(\log n)$ steps of finding extreme points until we are in a situation where $\overrightarrow{c_wc_v}$ lies in between $\overrightarrow{c_wu}$ and $\overrightarrow{c_wu'}$, in which case we can decide if $c$ in $R(v,w)$ (by cases $(i)$ and $(ii)$).

So far we have ignored the fact that $v$ and $w$ are squares and not points. It is easy to integrate this into the algorithm by dilating $R(v,w)$ with a fixed amount. A simple way of doing this is by using the four vectors $(\pm 1,0)^T$ and $(0,\pm 1)^T$ in the tree $\T(v,w)$ as extra vectors, effectively expanding the reachable region. Then we can keep on working with the centers $c_v$ and $c_w$ rather than the squares $v$ and~$w$.
\end{proof}

For the complete algorithm, we have $n$ starting nodes $v$, and for each, we traverse $S$. Treating any encountered node $w$ takes $O(\log^2 n)$ time for the search with $v$ in the region of $w$, and $O(\log n)$ time for updating $\T(v,w)$ to prepare for the next node $w'$ in the traversal of~$S$. Hence the overall algorithm runs in $O(n^2\log^2 n)$ time.
The analysis of the algorithm for the case where all couplings have size $O(1)$ or all couplings are purely horizontal or purely vertical is easy: The region $R(v,w)$ and the tree $\T(v,w)$ have constant complexity, so the logarithmic factors vanish. The running times improve to $O(n^2)$.
\medskip

When there are no couplings at all, we take a different approach that leads to a more efficient algorithm.
Choose an anchor node $v$ of the shape $S$ whose removal creates subtrees (subshapes) each of which has size at most half of the original size. Such a central node always exists, and we obtain at most four subtrees.
Observe that any operation in one subtree can influence the position of nodes in that subtree only, and not the positions of nodes in any other subtree. 

We will determine for each subtree separately where all of its nodes can be. If these regions overlap for any two subtrees, then we have two nodes in different subtrees of $S$ that can cause a collision and we have answered our question. 
We have also answered our question when one of the subtrees has a node that can occupy the location of the anchor $v$. Otherwise, we answer our question recursively in each subtree (note that the intersection of two regions of nodes in the same subtree does not tell us anything). In recursive steps, anchors are different and also the regions are different.

For each node $w_1$ in a subtree $S_1$, consider all operations on the path between $v$ and $w_1$. These are horizontal and vertical contractions and expansions. These together specify all locations where $w_1$ can possibly be. These locations necessarily form a rectangle $R(v,w_1)$ because there is no coupling. We do this for all nodes in $S_1$, giving a set ${\cal R}$ of ``red'' rectangles, and we do the same for all nodes in another subtree $S_2$, giving a set ${\cal B}$ of ``blue'' rectangles. We can compute ${\cal R}$ and ${\cal B}$ in linear time by tree traversal from the anchor $v$ and maintaining the operations on the path from~$v$.
With a standard plane sweep and segment trees, we can decide if any rectangle in ${\cal R}$ intersects any rectangle in ${\cal B}$ in $O(n \log n)$ time~\cite{bcko-cgaa-08}. Due to recursion on the subtrees of $S$ we spend $O(n\log^2 n)$ time overall.

\begin{theorem}
Let $S$ be a shape consisting of $n$ unit square nodes with operations defined on the edges between adjacent nodes, and let the adjacency structure of $S$ be a single tree. Then we can solve {\sc Colliding Schedule}
\begin{itemize}
\item in $O(n^2\log^2 n)$ time if couplings exist;
\item in $O(n^2)$ time if each coupling has constant size, or is horizontal-only or vertical-only;
\item in $O(n\log^2 n)$ time if the operations are not coupled.
\end{itemize}
\end{theorem}

We can also solve {\sc Discrete Colliding Schedule} in polynomial time. The algorithm without coupling is still correct, but with coupling we need a different approach.
We now present an algorithm to decide whether $w$ can be moved with respect to another node $v$ such that they collide in the discrete-time model and when there are couplings. Recall that the couplings on the path between $v$ and $w$ imply a set of vectors $\Vect(v,w)$. Unlike the continuous-time case, the reachable region $R(v,w)$ of $w$ is no longer convex when there are couplings.

Our approach is an extension of the classical dynamic-programming method to solve subset sum in the case where the numbers are integers of bounded value.
There are two main differences. Firstly, we do not have integers, but pairs of integers (vectors). Secondly, we need to handle the case where movement of the node $w$ collides with $v$ \emph{during} the movement.

We put the vectors of $\Vect(v,w)$ into an arbitrary order $\tau_1,\ldots,\tau_m$, and
for all $1\leq j\leq m$, we leave out vector $\tau_j$ and rename the others to become $\nu_1,\ldots,\nu_{m-1}$. We will show how to generate all cell locations that $w$ can reach using all subsets of these $\nu$ vectors, and then test if moving $w$ over vector $\tau_j$ causes a collision with $v$. Since every colliding schedule must end with some vector that causes the collision, we cover all possibilities.

We define a state $(i,x,y)$ as a Boolean that should be set to $\true$ by the algorithm if and only if some subset of $\nu_1,\ldots,\nu_i$ gives a sum of the first part that equals $x$ and sum of the second part that equals $y$.
We store the state in a 3-dimensional table which has size $O(n^3)$ as $m\leq n$ and the maximum values of $x$ and $y$ are~$n$.

In dynamic-programming fashion, we define $(i+1,x,y)$ based on entries with smaller first index.
We initialize the table with all entries set to $\false$. Then we set $(0,0,0):=\true$.
Let $\nu_{i+1}= \langle \nu^x_{i+1}, \nu^y_{i+1} \rangle$.

\begin{itemize}
	\item 
	If $(i,x,y)= \true$, then we set $(i+1,x,y):= \true$.
	\item
	If $(i,x,y)=\true$, then
	we set $(i+1,\,x+\nu^x_{i+1},\, y+\nu^y_{i+1}):=\true$.
\end{itemize}

When the table is filled, we know all possible starting locations of $w$ before the final move over the vector $\tau_j$. We test them all for a collision with $v$. In total, this algorithm takes $O(n^4)$ time.
This basic algorithm can be improved with the following observation:

\begin{lemma}
	The set $\Vect(v,w)$ has at most $O(n^{2/3})$ unique vectors.
\end{lemma}
\begin{proof}
	Any vector $\langle i,j\rangle$ ``costs'' at least $|i|+|j|$ nodes on the path between $v$ and $w$ in a single coupling. To make unique vectors, we need larger and larger couplings. It is clear that there cannot be more than $n^{2/3}$ vectors where one dimension is larger than $n^{1/3}$ in absolute value. At the same time, the number of vectors where both dimensions are at most $n^{1/3}$ in absolute value is $4n^{2/3}$.
\end{proof}

Since vector addition is commutative, we can assume that all vectors with the same direction are consecutive. This means that we can exclude not just one vector $\tau_j$, but all equal vectors. At the end we test collision with the summed vector in that direction. Since there are only $O(n^{2/3})$ unique vectors, we run the dynamic programming method only $O(n^{2/3})$ times, giving an running time of $O(n^{11/3})$ to test if $v$ and $w$ can intersect.

Moreover, when couplings have size $O(1)$ or we have horizontal-only and vertical-only couplings, we have only $O(1)$ directions of vectors, leading to $O(n^3)$  running time.
These running times are for each pair of nodes $v,w$, so overall we get:

\begin{theorem}
	Let $S$ be a shape consisting of $n$ unit square nodes with operations defined on the edges between adjacent nodes, and let the adjacency structure of $S$ be a single tree. Then we can solve {\sc Discrete Colliding Schedule}
	\begin{itemize}
		\item in $O(n^{17/3})$ time if couplings exist;
		\item in $O(n^5)$ time if each coupling has constant size, or is horizontal-only or vertical-only;
		\item in $O(n\log^2 n)$ time if the operations are not coupled.
	\end{itemize}
\end{theorem}

For the last result we can use the continuous-time model result, since the reachable region $R(v,w)$ is the same in the discrete and continuous cases when no couplings exist.

Finally, we note that when there is just one single coupling of all operations, then {\sc Colliding Schedule} and {\sc Collision-free Schedule} are the same, as are the discrete-time and continuous-time models. We can model the problem by adding a third dimension. Each node is at its starting location in the plane $z=0$ and at its destination---given some anchor---in the plane $z=1$. We make a prism with square horizontal cross-section for each node by connecting the starting and destination squares. The question whether a node collision occurs then boils down to determining whether a proper intersection occurs among the prisms. We can solve this problem, for any constant $\epsilon>0$, in $O(n^{3/2+\epsilon})$ time by ray shooting along all of the prism edges in the faces of the other prisms~\cite{ezra2022ray}.

\subsection{Colliding Schedule for General Graphs}
\label{section:general-shapes}

Recall that a cycle collision occurs between two nodes $u$ and $v$ on a cycle if the two paths between them do not agree on where one node moves to with respect to the other node. For any cycle that contains an individual operation in that cycle, it is easy to cause a cycle collision by performing that operation only. More generally, the horizontal operations must be in couplings that compensate for their effect on the cycle, and the same is true for the vertical operations. If this is not true for any cycle, then a cycle collision can be created. The possibility of a cycle collision is not dependent on the two nodes whose relative movement we check. Any pair of nodes on the cycle will give the same result.

We can show that if the couplings are such that in every \emph{elementary cycle} of the graph, there is no cycle collision, then this is true for the whole graph. So we need to check the $O(n)$ elementary cycles only. We make this more precise.

Let $S=(V,E)$ be the connectivity graph of a connected shape. An \emph{elementary cycle} is the outer boundary of a bounded face of a biconnected component. We obtain them by removing all nodes that are cut nodes of the graph. What remains are zero or more components each of which is biconnected. Each bounded face of the embedding is enclosed by a simple cycle of nodes. These are the elementary cycles. 

\begin{lemma}
\label{lem:cycles}
For any connected shape $S$, the elementary cycles cannot cause a cycle collision if and only if all cycles cannot cause a cycle collision.
\end{lemma}

\begin{proof}
	We need to prove only that non-collision of elementary cycles implies non-collision of any cycle.
	
	\begin{figure}[tb]
		\centering\includegraphics[width=6cm]{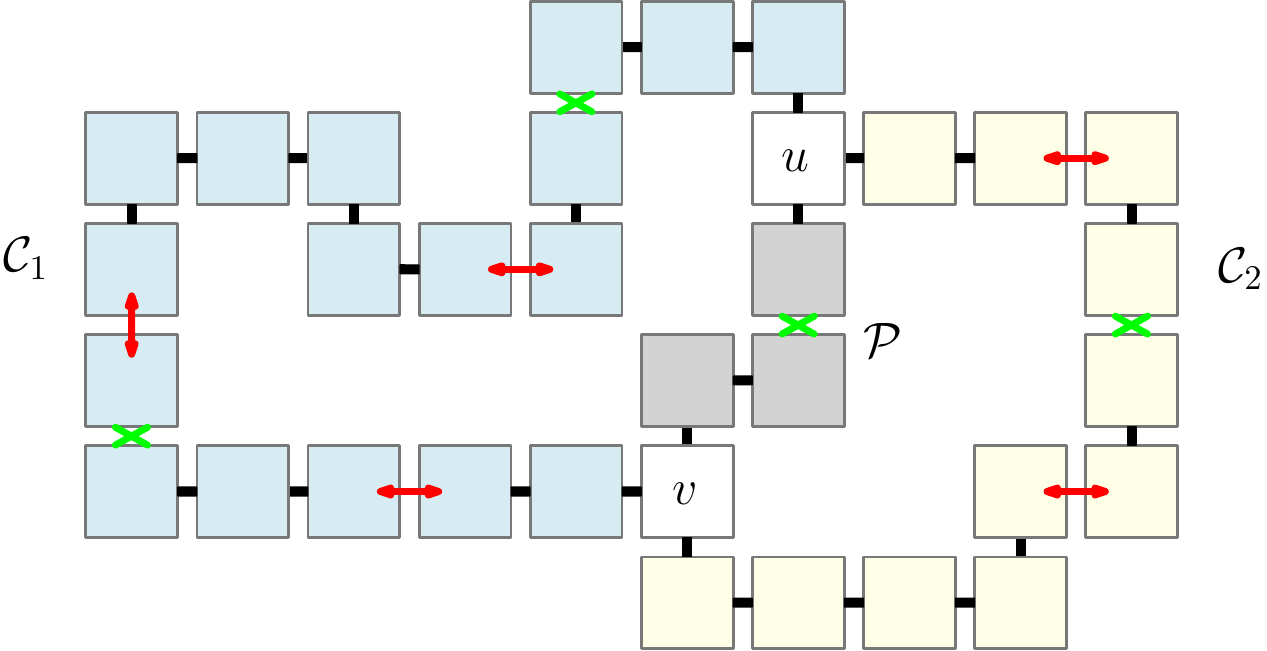}
		\caption{A simple cycle ${\cal C}$ split into ${\cal C}_1$ and ${\cal C}_2$ at two nodes $u$ and $v$, and an interior path ${\cal P}$ (of three nodes) connecting $u$ and $v$.}
		\label{fig:two-cycles}
	\end{figure}
	
	It is easy to see that in any simple cycle, a horizontal operation must be in a coupling with another horizontal operation in that cycle. By symmetry the same holds for vertical operations.
	It is also easy to see that operations incident to a cut node cannot cause a cycle collision, only a node collision.
	The latter observation leads to the fact that we can restrict ourselves to biconnected components.
	
	Let ${\cal C}$ be any cycle in a biconnected component, and assume that all elementary cycles of that biconnected component cannot cause cycle collisions. We will show that ${\cal C}$ cannot have a cycle collision either, by induction on the number $f$ of faces inside ${\cal C}$.
	
	If $f=1$, then ${\cal C}$ is elementary so the claim is true by assumption. Let $f>1$ and assume that all cycles with $\leq f-1$ faces inside cannot cause a cycle collision.
	Take a simple path ${\cal P}$ in the biconnected component and inside ${\cal C}$; see Figure~\ref{fig:two-cycles}. Let $u,v$ be the two nodes on ${\cal C}$ that are adjacent to the first and last nodes of ${\cal P}$. We consider the three paths between $u$ and $v$: two on ${\cal C}$, and ${\cal P}$ itself. We let $u$ and $v$ be the start and end of each of these paths, so that all operations on ${\cal C}$, on ${\cal P}$, and between ${\cal P}$ and $u$ and $v$ occur on exactly one of these paths. Call these three paths ${\cal C}_1$, ${\cal C}_2$, and ${\cal P}'$.
	
	By induction we know that the cycles formed by ${\cal C}_1$ and ${\cal P}$ and by  ${\cal C}_2$ and ${\cal P}$ cannot cause cycle collisions, since they have fewer than $f$ faces inside.
	Treat $u$ as an anchor, and consider how the operations on ${\cal C}_1$ move $v$. This is specified by a vector $\langle i,j\rangle$. Since the cycle formed by ${\cal C}_1$ and ${\cal P}$ cannot cause collisions, the path ${\cal P}'$ must also move $v$ over the same vector $\langle i,j\rangle$. Repeating this argument on ${\cal C}_2$ and ${\cal P}$, we see that ${\cal C}_2$ must also move $v$ by $\langle i,j\rangle$.
	Hence, the cycle ${\cal C}$ cannot cause a cycle collision either. The argument holds for any couplings that make sure that the smaller cycles do not cause cycle collisions.
\end{proof}

We now detect possible collisions as follows: First, we test if a cycle collision can be made in any elementary cycle. This is easy to do in linear time overall. If we found a collision, we are done. If not, we must check for node collisions. In Section~\ref{sec:model} we observed:

\begin{observation}
For a graph $S$ with cycles which cannot cause cycle collisions, the occupied cells after any subset of the operations is the same as for any spanning tree of $S$.
\end{observation}

Hence, we compute any spanning tree $T$ of $S$ and use the algorithm of Section~\ref{subsection:collisions-trees} on $T$ to find node collisions. 
Pairs of nodes between which an edge in $S$ was omitted in $T$ cannot cause node collisions. This extends the results of Section~\ref{subsection:collisions-trees} from trees to general connected graphs with the same time bounds.

\section{Continuous and Discrete Collision-free Schedule}
\label{sec:avoiding}

So far we considered detecting whether collisions might occur for an input instance. In this section, we consider the problem of deciding if all operations can be performed without any collisions, for a suitable choice of operation order or starting times.
We show that, even if there are only expansions that are horizontal and couplings have size $O(1)$, in the discrete-time model the problem is NP-complete. Interestingly, the same problem is solvable in polynomial time in the continuous-time model. When we add vertical expansions, the problem is NP-hard in the continuous-time model.

\begin{theorem}
\label{t:nphard1}
\textsc{Discrete Collision-free Schedule} is NP-complete even if all operations are horizontal expansions and all couplings have size $O(1)$.
\end{theorem}

\begin{proof}
In order to prove NP-hardness, we reduce the following problem to \textsc{Discrete Collision-free Schedule}.

\medskip

\noindent\textbf{{\sc LSAT (Linear SAT)}.} Given a CNF formula where each clause contains exactly three literals, each clause has common literals with at most one other clause, and two clauses have at most one common literal, decide if there is a satisfying assignment for the formula.

\medskip

\textsc{LSAT} is NP-complete \cite{DBLP:journals/dam/ArkinBCCKMS18}.
Let $\phi$ be a formula of \textsc{LSAT} with $n$ variables $\{ x_1,\dots,x_n \}$ and $m$ clauses.
We construct an instance of \textsc{Discrete Collision-free Schedule} for formula $\phi$ such that there is a collision-free schedule for the instance if and only if $\phi$ is satisfiable.

For each variable $x_i$, we utilize the \emph{variable gadget} shown in Figure~\ref{fig:gadgets}(a).
We perform operation $x_i$ or $\bar x_i$ to set variable $x_i$ to true or false, respectively.
The gadget causes a collision if we try to set variable $x_i$ both to true and false.
We use operation $y_i$ to enforce that the variable is set, and later use operation $z_i$ to make sure that the other operation, $\bar x_i$ or $x_i$, can also be performed.

\begin{figure}[tb]
\centering
\includegraphics[scale=0.95]{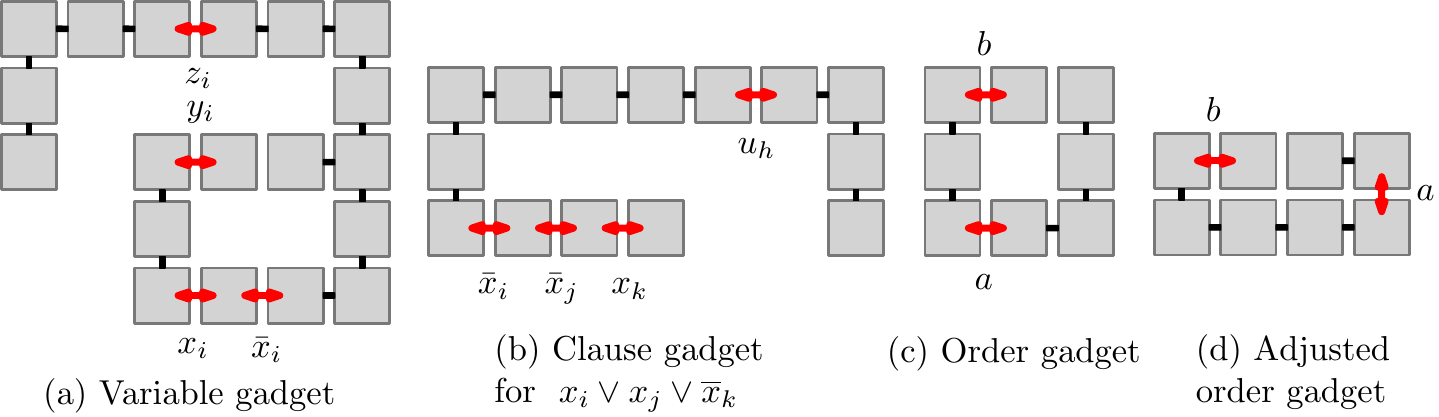}
\caption{Gadgets for the NP-hardness proof of {\sc Collision-free Schedule}.
	The gadget in (c) is used for the discrete case; the continuous case requires the gadget shown in (d).}
\label{fig:gadgets}
\end{figure}

For each clause, we utilize the \emph{clause gadget} shown in Figure~\ref{fig:gadgets}(b).
The gadget contains an operation for each literal that we couple with its inverse in the respective variable gadget, i.e., the operation is performed if and only if the literal is set to false.
The gadget causes a collision if and only if all three operations are performed, i.e., all three literals are set to false; here we assume that operation $u_h$ cannot be performed yet.
Hence, a collision implies that the clause is not satisfied.
We find a satisfying assignment for $\phi$ if and only if there are no collisions after setting all variables.

Two more ingredients are needed: First, we must make sure that all operations are eventually possible if $\phi$ can be satisfied, because \textsc{Discrete Collision-free Schedule} requires that we perform all operations. For that, we have included a release operation in each variable and clause gadget, namely $z_i$ resp.\ $u_h$. By expanding those, we obtain enough space to expand all $x_i$s and $\bar x_i$s.

Second, we must ensure that the release operations are not possible before all variables are set.
For that, we introduce the \emph{order gadget} shown in Figure~\ref{fig:gadgets}(c).
The gadget causes a collision if and only if we perform operation $b$ before $a$.
Hence, it forces an order between $a$ and $b$.
Let $a \prec b$ denote that $a$ has to be performed before $b$.
We define the following order between the operations:
\[y_1 \prec y_2 \prec \dots \prec y_n \prec z_1 \prec z_2 \prec \dots \prec z_n \prec u_1 \prec u_2 \prec \dots u_m.\]
We need a single order gadget for each relation in the chain, $2n+m-1$ in total.
At the time we perform $y_n$, each variable must have been set to either true or false. This is the moment when a truth assignment of $\phi$ is tested.
Note that in the discrete-time model, each coupled set of operations must have completed before the next one starts, so we cannot start with any $z$ or $u$ before all $y$ are completed, which requires that each variable has been set. 

We have only utilized horizontal expansions in our gadgets.
Furthermore, in \textsc{LSAT}, each literal occurs at most twice.
Hence, each operation $x_i$ and $\bar x_i$ occurs once in a variable gadget, and at most twice in clause gadgets.
Each other operation occurs once in a variable or clause gadget, and at most twice in order gadgets.
Thus, each coupling is of constant size. We can connect all gadgets into a single tree structure of linear size to produce a connected shape~$S$.

Finally, note that \textsc{Discrete Collision-free Schedule} is in NP since we can trivially check any schedule for collisions in polynomial time.
\end{proof}

We switch to the continuous-time model and give a positive and a negative result that complement the NP-completeness result in the discrete-time model. 

\begin{theorem}
\label{th:cccs:polynomial}
\textsc{Collision-free Schedule} is solvable in linear time if all operations are horizontal.
\end{theorem}

\begin{proof}
We will show that there is a collision-free schedule if and only if the execution of all operations at once is collision-free.
Notice that couplings are irrelevant: all operations will be performed simultaneously in any case, which is allowed regardless of the presence of couplings in the continuous-time model. Each node moves with a fixed speed from its initial position to its destination.

Observe that if all operations are horizontal, no node can leave its row. Hence, we can set up constraints for avoiding collisions for every row independently. Consider any \emph{gap} in any row of a shape $S$: a bounded maximal subrow of empty grid cells. Operations may shrink or grow a gap by one unit, other operations will not affect it. The non-collision constraint for a gap states that the number of operations that grow the gap plus the initial size of the gap is at least as large as the number of operations that shrink the gap. There is no collision if and only if no gap gets negative size.

We can determine the pairs of nodes that enclose each gap in linear time by traversing along the outside of the (tree) shape and maintaining a single stack. We can determine where each node ends up in linear time overall as well, giving us the constraints.
\end{proof}

The proof can be seen as a study of the feasible region in a $k$-dimensional space when there are $k$ operations. The gap constraints are linear constraints. If the starting shape has no collisions and the final shape has no collisions, then these shapes are both in the feasible region, and since the feasible region is convex, the connecting segment between them is in the feasible region as well. This connecting segment in $k$-dimensional space corresponds to executing all operations in parallel.

\begin{theorem}
\label{t:nphard2}
\textsc{Collision-free Schedule} is NP-hard.
\end{theorem}

\begin{proof}
In order to prove NP-hardness for \textsc{Collision-free Schedule}, we adjust the order gadget such that it does not only force an order between operations but also separates them in time (see Figure~\ref{fig:gadgets}(d)).
For that, it uses a horizontal and a vertical expansion.
\end{proof}

\section{Conclusions and Open Problems}

We have studied algorithmic questions associated with modular reconfigurable robot models which can be summarized as ``Can collisions be caused?'' and ``Can collisions be avoided?'', while performing all operations in the specification. All versions of causing collisions are polynomial-time solvable, but for avoiding collisions, this is not true. 

The most intriguing open problem is whether the general version of avoiding collisions without coupling is NP-hard, or whether it can be solved in polynomial time. It would also be interesting to know if causing collisions can be solved in subquadratic time with couplings.

\clearpage

\bibliography{main}
\end{document}